\documentclass[11pt]{amsart}
\usepackage{amsthm}
\usepackage{amsmath,amstext,amsthm,amsfonts,amssymb,amsthm}
\usepackage{multicol}
\usepackage{latexsym}
\usepackage{color}
\usepackage{graphicx}
\usepackage{epsf}
\usepackage{epsfig}
\usepackage{epic}
\usepackage{graphics}
\usepackage{verbatim}
\usepackage{color}
\usepackage{stmaryrd}
\usepackage{hyperref}
\newtheorem{theorem}{Theorem}[section]
\newtheorem{lemma}[theorem]{Lemma}
\theoremstyle{definition}
\newtheorem{definition}[theorem]{Definition}
\newtheorem{proposition}[theorem]{Proposition}

\theoremstyle{remark}
\newtheorem{remark}[theorem]{Remark}
\numberwithin{equation}{section}

\newcommand{\quat}{\mathbb H}

\newcommand{\an}{\mathfrak{a}}

\newcommand{\q}{\mathbf{q}}
\newcommand{\I}{\mathbb{I}}
\newcommand{\HI}{\mathfrak{H}}

\newcommand{\D}{\mathfrak{D}}
\newcommand{\C}{\mathbb{C}}

\newcommand{\qu}{\mathfrak{q}}
\newcommand{\pu}{\mathfrak{p}}
\newcommand{\as}{\mathsf{a}}
\newcommand{\asd}{\mathsf{a}^{\dagger}}

\newcommand{\oqu}{\overline{\mathfrak{q}}}
\newcommand{\opu}{\overline{\mathfrak{p}}}
\newcommand{\HIB}{\mathfrak{H}^B_r}
\newcommand{\Ci}{\mathfrak{C}_i}

\begin{document}
\title[Squeezed states ]{ Canonical, squeezed and fermionic coherent states in a right quaternionic Hilbert space with a left multiplication on it}
\author{ K. Thirulogasanthar$^{\dagger}$, B. Muraleetharan$^{\ddagger}$}
\address{$^{\ddagger}$ Department of Computer Science and Software Engineering, Concordia University, 1455 De Maisonneuve Blvd. West, Montreal, Quebec, H3G 1M8, Canada.}
\address{$^{\ddagger}$ Department of mathematics and Statistics, University of Jaffna, Thirunelveli, Sri Lanka.}
\email{santhar@gmail.com, bbmuraleetharan@jfn.ac.lk}
\date{\today}
\thanks{K. Thirulogasanthar would like to thank the, FQRNT, Fonds de la Recherche  Nature et  Technologies (Quebec, Canada) for partial financial support Under the grant number 2017-CO-201915. Authors would like to thank I. Sabadini for discussions.}
\pagestyle{myheadings}
\dedicatory{Dedicated to the memory of S. Twareque Ali}
\begin{abstract}
  Using a left multiplication defined on a right quaternionic Hilbert space, we shall demonstrate that various classes of coherent states such as the canonical coherent states, pure squeezed states, fermionic coherent states can be defined with all the desired properties on a right quaternionic Hilbert space. Further, we shall also  demonstrate squeezed states can be defined on the same Hilbert space, but the noncommutativity of quaternions prevents us in getting the desired results.
\end{abstract}

\maketitle

\section{Introduction}\label{sec_intro}
Quantum mechanics can be formulated over the complex and the quaternionic numbers, see \cite{bvn, Ad, Ali}. In recent times, new mathematical tools in quaternionic analysis became available in the literature. In particular the spectral theory and the functional calculus. As a consequence, there has been a renewed interest in the quaternionic quantum mechanics.
As in the complex quantum mechanics,  states  are represented by vectors of a separable quaternionic Hilbert space and observables are represented by quaternionic linear and self-adjoint operators\cite{Ad}

Until the most recent times, an appropriate spectral theory was missing since there was not a satisfactory notion of spectrum. This difficulty has been solved with the introduction of the notion of S-spectrum (see \cite{NFC}) and, accordingly, with a proof of the spectral theorem for normal operators, see \cite{ack}.

Sine for any linear operator $A$ and a quaternion $\qu\in\quat$, $(\qu A)^{\dagger}\not=\oqu A^{\dagger}$. Due to this, in quaternionic quantum mechanics there is no proper momentum operator. For various attempts and their drawbacks we refer the reader to \cite{Ad}. However, in \cite{MTI}, we have given a through discussion on the possibility of defining various momentum operators. In fact we have shown that, by using the notion of left multiplication in a right quaternionic Hilbert space it is possible to define a linear self-adjoint momentum operator in complete analogy with the complex case.

With the right multiplication on a right quaternion Hilbert space a displacement operator similar to the harmonic oscillator displacement operator cannot be defined as a representation of the Fock space \cite{Ad2, Thi2}. However, in \cite{MTI} we have shown that, with the aid of a left multiplication defined on a right quaternionic Hilbert space, an appropriate harmonic oscillator displacement operator can be defined. We have proved that this operator is square integrable, irreducible and a unitary representation and also it satisfies most of the properties of its complex counterpart.

In \cite{TM} we introduce and studied the squeeze operator which is formally defined as in the complex setting but with the left multiplication on a right quaternion Hilbert space. We have shown, in the quaternion setting, that pure squeezed states can be obtained with all the desired properties. However,
due to the non-commutative nature of quaternions, there is an intrinsic issue if one is aimed to obtain relations involving both the displacement and the squeeze operator. Suitable relations can be obtained only quaternion slice-wise.

This article is written as a review article, except few new materials, indicating all these aspects with appropriate references.
The plan of the paper is as follows. Section 2 contains some preliminaries on quaternions, right quaternionic Hilbert spaces and the notion  of left multiplication. Section 3 studies the Bargmann space of regular functions, the displacement operator, the squeeze operator and some of its properties. The expectation values and the variances of the creation and annihilation operator and of the quadrature operators are computed in this section in terms of pure squeezed states. The fourth section is devoted to the relation involving displacement, squeeze operators to obtain squeezed states on a quaternion slice. We obtain a result similar to the one in the complex case at the end of section 3 on the whole set of quaternions, but due the noncommutativity  expectation values and variances can be computed only on quaternionic slices which is also demonstrated in section 4. Section five deals with fermionic states for one mode systems which is not reported elsewhere. Section six ends the manuscript with a conclusion.

\section{Mathematical preliminaries}
In this section we recall some basic facts about quaternions, their complex matrix representation, quaternionic Hilbert spaces as needed here. For details we refer the reader to \cite{Ad, Vis, Za, Al1}.
\subsection{Quaternions}
Let $\quat$ denote the field of quaternions. Its elements are of the form $\qu=q_0+q_1i+q_2j+q_3k$ where $q_0,q_1,q_2$ and $q_3$ are real numbers, and $i,j,k$ are imaginary units such that $i^2=j^2=k^2=-1$, $ij=-ji=k$, $jk=-kj=i$ and $ki=-ik=j$. The quaternionic conjugate of $\qu$ is defined to be $\overline{\qu} = q_0 - q_1i - q_2j - q_3k$. Quaternions can be represented by $2\times 2$ complex matrices:
\begin{equation}
\qu = q_0 \sigma_{0} + i \q \cdot \underline{\sigma},
\end{equation}
with $q_0 \in \mathbb R , \quad \q = (q_1, q_2, q_3)
\in \mathbb R^3$, $\sigma_0 = \mathbb{I}_2$, the $2\times 2$ identity matrix, and
$\underline{\sigma} = (\sigma_1, -\sigma_2, \sigma_3)$, where the
$\sigma_\ell, \; \ell =1,2,3$ are the usual Pauli matrices. The quaternionic imaginary units are identified as, $i = \sqrt{-1}\sigma_1, \;\; j = -\sqrt{-1}\sigma_2, \;\; k = \sqrt{-1}\sigma_3$. Thus,
\begin{equation}
\qu = \left(\begin{array}{cc}
q_0 + i q_3 & -q_2 + i q_1 \\
q_2 + i q_1 & q_0 - i q_3
\end{array}\right) \qquad 
\label{q3}
\end{equation}
and $\overline{\qu} = \qu^\dag\quad \text{(matrix adjoint)}\; .$
Using the polar coordinates:
\begin{eqnarray*}
	q_0 &=& r \cos{\theta}, \\
	q_1 &=& r \sin{\theta} \sin{\phi} \cos{\psi}, \\
	q_2 &=& r \sin{\theta} \sin{\phi} \sin{\psi}, \\
	q_3 &=& r \sin{\theta} \cos{\phi},
\end{eqnarray*}
where $(r,\phi,\theta,\psi)  \in [0,\infty)\times[0,\pi]\times[0,2\pi)^{2}$, we may write
\begin{equation}
\qu = A(r) e^{i \theta \sigma(\widehat{n})}
\label{q4},
\end{equation}
where
\begin{equation}
A(r) = r\mathbb \sigma_0
\end{equation}
and\begin{equation}
\sigma(\widehat{n}) = \left(\begin{array}{cc}
\cos{\phi} & \sin{\phi} e^{i\psi} \\
\sin{\phi} e^{-i\psi} & -\cos{\phi}
\end{array}\right).
\label{q5}\end{equation}
The matrices
$A(r)$ and $\sigma(\widehat{n})$ satisfy the conditions,
\begin{equation}
A(r) = A(r)^\dagger,~\sigma(\widehat{n})^2 = \sigma_0
,~\sigma(\widehat{n})^\dagger = \sigma(\widehat{n})
\label{san1}
\end{equation}
and
$\lbrack A(r), \sigma(\widehat{n}) \rbrack = 0.$
Note that a real norm on $\quat$ is defined by  $$\vert\qu\vert^2  := \overline{\qu} \qu = r^2 \sigma_0 = (q_0^2 +  q_1^2 +  q_2^2 +  q_3^2).$$
Note also that for ${\pu},\qu\in \quat$, we have $\overline{{\pu}\qu}=\overline{\qu}~\overline{{\pu}}$, $\pu\qu\not=\qu\pu$, $\qu\overline{{ \qu}}=\overline{{\qu}}\qu$, and real numbers commute with quaternions.
Quaternions can also be interpreted as a sum of a scalar and a vector by writing $$\qu=q_0+q_1i+q_2j+q_3k=(q_0,\q);$$ where $\q=q_1i+q_2j+q_3k$.
We borrow the materials as needed here from \cite{Al1}.  Let
\begin{eqnarray*}
	\mathbb{S}&=&\{I=x_1i+x_2j+x_3k~\vert
	~x_1,x_2,x_3\in\mathbb{R},~x_1^2+x_2^2+x_3^2=1\},
\end{eqnarray*}
we call it a quaternion sphere.
 \begin{proposition}\cite{Al1}\label{Pr1}
	For any non-real quaternion $\qu\in \quat\smallsetminus\mathbb{R}$, there exist, and are unique, $x,y\in\mathbb{R}$ with $y>0$, and $I_\qu\in\mathbb{S}$ such that $\qu=x+I_\qu y$.
\end{proposition}
For every quaternion $I\in\mathbb{S}$, the complex plane $\C_I=\mathbb{R}+I\mathbb{R}$ passing through the origin, and containing $1$ and $I$, is called a quaternion slice. Thereby, we can see that
\begin{equation}\label{Eq1}
\quat=\bigcup_{I\in\mathbb{S}}\C_I\quad\text{and}\quad\bigcap_{I\in\mathbb{S}} \C_I=\mathbb{R}
\end{equation}
One can also easily see that $\C_I\subset \quat$ is commutative, while, elements from two different quaternion slices, $\C_I$ and $\C_J$ (for $I, J\in\mathbb{S}$ with $I\not=J$), do not necessarily commute.
\subsection{Quaternionic Hilbert spaces}
In this subsection we introduce right quaternionic Hilbert spaces. For details we refer the reader to \cite{Ad}. We also define the Hilbert space of square integrable functions on quaternions based on \cite{Vis, Gu,Al1}.
\subsubsection{Right Quaternionic Hilbert Space}
Let $V_{\quat}^{R}$ be a linear vector space under right multiplication by quaternionic scalars (again $\quat$ standing for the field of quaternions).  For $f,g,h\in V_{\quat}^{R}$ and $\qu\in \quat$, the inner product
$$\langle\cdot\mid\cdot\rangle:V_{\quat}^{R}\times V_{\quat}^{R}\longrightarrow \quat$$
satisfies the following properties
\begin{enumerate}
	\item[(i)]
	$\overline{\langle f\mid g\rangle}=\langle g\mid f\rangle$
	\item[(ii)]
	$\|f\|^{2}=\langle f\mid f\rangle>0$ unless $f=0$, a real norm
	\item[(iii)]
	$\langle f\mid g+h\rangle=\langle f\mid g\rangle+\langle f\mid h\rangle$
	\item[(iv)]
	$\langle f\mid g\qu\rangle=\langle f\mid g\rangle\qu$
	\item[(v)]
	$\langle f\qu\mid g\rangle=\overline{\qu}\langle f\mid g\rangle$
\end{enumerate}
where $\overline{\qu}$ stands for the quaternionic conjugate. We assume that the
space $V_{\quat}^{R}$ is complete under the norm given above. Then,  together with $\langle\cdot\mid\cdot\rangle$ this defines a right quaternionic Hilbert space, which we shall assume to be separable. Quaternionic Hilbert spaces share most of the standard properties of complex Hilbert spaces. In particular, the Cauchy-Schwartz inequality holds on quaternionic Hilbert spaces as well as the Riesz representation theorem for their duals.  Thus, the Dirac bra-ket notation
can be adapted to quaternionic Hilbert spaces:
$$\mid f\qu\rangle=\mid f\rangle\qu,\hspace{1cm}\langle f\qu\mid=\overline{\qu}\langle f\mid\;, $$
for a right quaternionic Hilbert space, with $\vert f\rangle$ denoting the vector $f$ and $\langle f\vert$ its dual vector. Similarly the left quaternionic Hilbert space $V_{\quat}^{L}$ can also be described, see for more detail \cite{Ad,MuTh,Thi1}.
The field of quaternions $\quat$ itself can be turned into a left quaternionic Hilbert space by defining the inner product $\langle \qu \mid \qu^\prime \rangle = \qu \qu^{\prime\dag} = \qu\overline{\qu^\prime}$ or into a right quaternionic Hilbert space with  $\langle \qu \mid \qu^\prime \rangle = \qu^\dag \qu^\prime = \overline{\qu}\qu^\prime$. Further note that, due to the non-commutativity of quaternions the sum
$\sum_{m=0}^{\infty}\pu^m\qu^m/m!$
cannot be written as $\text{exp}(\pu\qu).$ However, in any Hilbert space the norm convergence implies the convergence of the series and
$\sum_{m=0}^{\infty}\left\vert \pu^m\qu^m/m!\right\vert
\leq e^{|\pu||\qu|},$ therefore $\sum_{m=0}^{\infty}\pu^m\qu^m/m!=e^{\pu\qu}_*$ converges, where $e^{\pu\qu}_*$ is given as a quaternion star product \cite{NFC}
\subsubsection{Quaternionic Hilbert Spaces of Square Integrable Functions}
Let $(X, \mu)$ be a measure space and $\quat$  the field of quaternions, then
$$L^2_\quat(X, d\mu)=\left\{f:X\rightarrow \quat \left|  \int_X|f(x)|^2d\mu(x)<\infty \right.\right\}$$
\noindent
is a right quaternionic Hilbert space which is denoted by $L^2_\quat(X,\mu)$, with the (right) scalar product
\begin{equation}
\langle f \mid g\rangle =\int_X\overline{ f(x)}{g(x)} d\mu(x),
\label{left-sc-prod}
\end{equation}
where $\overline{f(x)}$ is the quaternionic conjugate of $f(x)$, and (right)  scalar multiplication $f\an, \; \an\in \quat,$ with $(f\an)(\qu) = f(\qu)\an$ (see \cite{Gu,Vis} for details). Similarly, one could define a left quaternionic Hilbert space of square integrable functions.
\subsection{Left Scalar Multiplications on $V_{\quat}^{R}$.}
We shall extract the definition and some properties of left scalar multiples of vectors on $V_{\quat}^R$ from \cite{ghimorper} as needed for the development of the manuscript. The left scalar multiple of vectors on a right quaternionic Hilbert space is an extremely non-canonical operation associated with a choice of preferred Hilbert basis. Now the Hilbert space $V_{\quat}^{R}$ has a Hilbert basis
\begin{equation}\label{b1}
\mathcal{O}=\{\varphi_{k}\,\mid\,k\in N\},
\end{equation}
where $N$ is a countable index set.
The left scalar multiplication `$\cdot$' on $V_{\quat}^{R}$ induced by $\mathcal{O}$ is defined as the map $\quat\times V_{\quat}^{R}\ni(\qu,\phi)\longmapsto \qu\cdot\phi\in V_{\quat}^{R}$ given by
\begin{equation}\label{LPro}
\qu\cdot\phi:=\sum_{k\in N}\varphi_{k}\qu\langle \varphi_{k}\mid \phi\rangle,
\end{equation}
for all $(\qu,\phi)\in\quat\times V_{\quat}^{R}$. Since all left multiplications are made with respect to some basis, assume that the basis $\mathcal{O}$ given by (\ref{b1}) is fixed.
\begin{proposition}\cite{ghimorper}\label{lft_mul}
	The left product defined in (\ref{LPro}) satisfies the following properties. For every $\phi,\psi\in V_{\quat}^{R}$ and $\pu,\qu\in\quat$,
	\begin{itemize}
		\item[(a)] $\qu\cdot(\phi+\psi)=\qu\cdot\phi+\qu\cdot\psi$ and $\qu\cdot(\phi\pu)=(\qu\cdot\phi)\pu$.
		\item[(b)] $\|\qu\cdot\phi\|=|\qu|\|\phi\|$.
		\item[(c)] $\qu\cdot(\pu\cdot\phi)=(\qu\pu\cdot\phi)$.
		\item[(d)] $\langle\overline{\qu}\cdot\phi\mid\psi\rangle
		=\langle\phi\mid\qu\cdot\psi\rangle$.
		\item[(e)] $r\cdot\phi=\phi r$, for all $r\in \mathbb{R}$.
		\item[(f)] $\qu\cdot\varphi_{k}=\varphi_{k}\qu$, for all $k\in N$.
	\end{itemize}
\end{proposition}
\begin{remark}\label{RR1}
 It is immediate that $(\pu+\qu)\cdot\phi=\pu\cdot\phi+\qu\cdot\phi$, for all $\pu,\qu\in\quat$ and $\phi\in V_{\quat}^{R}$. Moreover, with the aid of (b) in above Proposition (\ref{lft_mul}), we can have, if $\{\phi_n\}$ in $V_\quat^R$ such that $\phi_n\longrightarrow\phi$, then $\qu\cdot\phi_n\longrightarrow\qu\cdot\phi$. Also if $\sum_{n}\phi_n$ is a convergent sequence in $V_\quat^R$, then $\qu\cdot(\sum_{n}\phi_n)=\sum_{n}\qu\cdot\phi_n$.
\end{remark}
Furthermore, the quaternionic scalar multiplication of $\quat$-linear operators is also defined in \cite{ghimorper}. For any fixed $\qu\in\quat$ and a given right $\quat$-linear operator $A:\D(A)\longrightarrow V_{\quat}^{R}$, the left scalar multiplication `$\cdot$' of $A$ is defined as a map $\qu \cdot A:\D(A)\longrightarrow V_{\quat}^{R}$ by the setting
\begin{equation}\label{lft_mul-op}
(\qu\cdot A)\phi:=\qu\cdot (A\phi)=\sum_{k\in N}\varphi_{k}\qu\langle \varphi_{k}\mid A\phi\rangle,
\end{equation}
for all $\phi\in \D(A)$. It is straightforward that $\qu A$ is a right $\quat$-linear operator. If $\qu\cdot\phi\in \D(A)$, for all $\phi\in \D(A)$, one can define right scalar multiplication `$\cdot$' of the right $\quat$-linear operator $A:\D(A)\longrightarrow V_{\quat}^{R}$ as a map $ A\cdot\qu:\D(A)\longrightarrow V_{\quat}^{R}$ by the setting
\begin{equation}\label{rgt_mul-op}
(A\cdot\qu )\phi:=A(\qu\cdot \phi),
\end{equation}
for all $\phi\in \D(A)$. It is also right $\quat$-linear operator. One can easily obtain that, if $\qu\cdot\phi\in \D(A)$, for all $\phi\in \D(A)$ and $\D(A)$ is dense in $V_{\quat}^{R}$, then
\begin{equation}\label{sc_mul_aj-op}
(\qu\cdot A)^{\dagger}=A^{\dagger}\cdot\overline{\qu}~\mbox{~and~}~
(A\cdot\qu)^{\dagger}=\overline{\qu}\cdot A^{\dagger}.
\end{equation}

\section{Bargmann space of regular functions}
The Bargmann space of left regular functions $\HI^B_r$ is a closed subspace of the right Hilbert space $L_{\quat}(\quat, d\zeta(r, \theta, \phi, \psi))$, where $d\zeta(r, \theta, \phi, \psi)=\frac{1}{4\pi} e^{-r^2}\sin{\phi} dr d\theta d\phi d\psi$. An orthonormal basis of this space is given by the monomials (which are both left and right regular)
$$\Phi_n(\qu)=\frac{\qu^n}{\sqrt{n!}};\quad n=0,1,2,\cdots.$$
There is also an associated reproducing kernel
$$K_B(\qu,\overline{\pu})=\sum_{n=0}^{\infty}\Phi_n(\qu)\overline{\Phi_n(\pu)}=e_{\star}^{\qu\overline{\pu}}$$
 see \cite{Thi1, Al} for details.

\subsection{Coherent states on right quaternionic Hilbert spaces}\label{CSLQH}
The main content of this section is extracted from \cite{Thi2} as needed here. For an enhanced explanation we refer the reader to \cite{Thi2}. In \cite{Thi2} the authors have defined coherent states on $V_{\quat}^{R}$ and $V_{\quat}^{L}$, and also established the normalization and resolution of the identities for each of them.\\

On the Bargmann space $\HI^B_r$, the normalized canonical coherent states are
\begin{equation}\label{CCS}
\eta_{\qu}=\frac{1}{\sqrt{K_B(\qu, \oqu)}}\sum_{n=0}^{\infty}\Phi_n\Phi_n(\oqu)=e^{-\frac{|\qu|^2}{2}}\sum_{n=0}^{\infty}\Phi_n\frac{\qu^n}{n!}
=e^{-\frac{|\qu|^2}{2}}\sum_{n=0}^{\infty}\frac{\qu^n}{n!}\cdot\Phi_n,
\end{equation}
where we have used the fact in Proposition \ref{lft_mul} (f), with a resolution of the identity
\begin{equation}\label{resolution}
\int_{\quat}|\eta_{\qu}\rangle\langle\eta_{\qu}|d\zeta(r, \theta, \phi, \psi)=I_{\HI^B_r}.
\end{equation}
Now take the corresponding annihilation and creation operators as
\begin{eqnarray*}
\as\Phi_0=0,\quad
\as\Phi_n=\sqrt{n}\Phi_{n-1},\quad
\asd\Phi_n=\sqrt{n+1}\Phi_{n+1}.
\end{eqnarray*}
The operators can be taken as $\asd=\qu$ (multiplication by $\qu$) and $\as=\partial_s$ (left slice regular derivative), see \cite{Thi1, MuTh}. It is also not difficult to see that $(\asd)^{\dagger}=\as$, $[\as,\asd]=I_{\HI^B_r}$ and $\as\eta_{\qu}=\qu\cdot\eta_{\qu}$ (see also \cite{MTI}). In the same way canonical CS can also be defined on a left quaternion Hilbert space \cite{Thi2}.\\
In the following we shall briefly see the Heisenberg uncertainty relation. The material are extracted from \cite{MuTh, MTI} and for an enhanced explanation we refer the reader to \cite{MuTh, MTI}. Through the coherent state quantization process the annihilation and creation operators can also be written as
$$\as=\sum_{n=0}^{\infty}\sqrt{n+1}|\Phi_n\rangle\langle\Phi_{n+1}|\quad\text{and}\quad
\asd=\sum_{n=0}^{\infty}\sqrt{n+1}|\Phi_{n+1}\rangle\langle\Phi_{n}|.$$
Let $N=\asd\as$ the number operator. For $\qu\in\quat$ let the position and the momentum coordinates as
$$q=\frac{1}{\sqrt{2}}(\qu+\oqu)\quad\text{and}\quad p=-\frac{i}{\sqrt{2}}(\qu-\oqu).$$
\begin{remark}Through linearity in the quantization if we take the momentum operator as $\displaystyle P= -\frac{i}{\sqrt{2}}(\as-\asd)$ then $P$ is not self-adjoint \cite{MuTh}. However, if we consider the left multiplication of operators then $P$ becomes a linear self-adjoint operator.
\end{remark}
\begin{proposition}\cite{MTI} The operators $\displaystyle Q=\frac{1}{\sqrt{2}}(\as+\asd)$ and $\displaystyle P= -\frac{i}{\sqrt{2}}\cdot(\as-\asd)$ are linear and self-adjoint. Further the $i$ in $P$ can be replaced by $j, k$ or any other $I\in\mathbb{S}$.
\end{proposition}
The operators $Q$ and $P$ are the quaternionic position and momentum operators respectively. The operator
$$H=\frac{Q^2+P^2}{2}=N+\frac{1}{2}I_{\HIB}$$
is the quaternionic analogue of the harmonic oscillator Hamiltonian. Now using the canonical coherent states $\eta_\qu$ we can compute the following expectation values.
\begin{eqnarray*}
\langle\eta_\qu|\as|\eta_\qu\rangle&=&\qu,\quad \langle\eta_\qu|\asd|\eta_\qu\rangle=\oqu,\\
\langle\eta_\qu|\as^2|\eta_\qu\rangle&=&\qu^2,\quad \langle\eta_\qu|(\asd)^2|\eta_\qu\rangle=\oqu^2,\\
\langle\eta_\qu|\as\asd|\eta_\qu\rangle&=&1+|\qu|^2\quad\text{and}\quad \langle\eta_\qu|\asd\as|\eta_\qu\rangle=|\qu|^2.
\end{eqnarray*}
Using these we obtain
\begin{eqnarray*}
\langle\eta_\qu|Q|\eta_\qu\rangle&=&\frac{1}{2}(\qu^2+2|\qu|^2+\oqu^2)\quad\text{and}\\
\langle\eta_\qu|Q^2|\eta_\qu\rangle&=&\frac{1}{2}(\qu^2+1+2|\qu|^2+\oqu^2).
\end{eqnarray*}
Hence we get
$$\langle\Delta Q\rangle^2=\langle\eta_\qu|Q^2|\eta_\qu\rangle-\langle\eta_\qu|Q|\eta_\qu\rangle^2=\frac{1}{2}$$
That is
$$\langle\Delta Q\rangle=\frac{1}{\sqrt{2}}.$$
However, due to the noncommutativity of quaternions there is a technical difficulty in computing $\langle\Delta P\rangle$. As we can see
$$\langle\eta_\qu|i\cdot\as|\eta_\qu\rangle=\left(e^{-|\qu|^2}\sum_{n=0}^{\infty}\frac{\oqu^n i\qu^n}{n!}\right)\qu=\Ci\qu\quad\text{(say)},$$
where $\displaystyle\Ci=e^{-|\qu|^2}\sum_{n=0}^{\infty}\frac{\oqu^n i\qu^n}{n!}$ cannot be computed explicitly and there is no known technique to overcome this difficulty. However, this series absolutely converges to $1$.
That is $|\Ci |\leq 1$.
Since $\overline{\Ci}=-\Ci$ and $|\Ci|^2=-\Ci^2$
we can take $\Ci=rI$ for some $r\in [0,1]$ and $I\in\mathbb{S}$. With this $\Ci$ we compute the following.
\begin{eqnarray*}
\langle\eta_\qu|P|\eta_\qu\rangle&=&\frac{1}{\sqrt{2}}(\Ci\qu-\Ci\oqu)\\
\langle\eta_\qu|P^2|\eta_\qu\rangle&=&-\frac{1}{2}(\qu^2-1-2|\qu|^2+\oqu^2)
\end{eqnarray*}
Hence
$$\langle\Delta P\rangle^2=-\frac{1}{2}(\qu^2-1-2|\qu|^2+\oqu^2)
-\frac{1}{2}[(\Ci\qu)^2+2|\Ci\qu|^2+\overline{(\Ci\qu)}^2].$$
Therefore, as $|\Ci|\leq 1$, we get
$$|\langle\Delta Q\rangle^2\langle\Delta P\rangle^2|\geq \frac{1}{4}-|\qu|^2.$$
Similarly
$$|\langle\Delta Q\rangle^2\langle\Delta P\rangle^2|\leq \frac{1}{4}+|\qu|^2.$$
That is,
$$|\langle\Delta Q\rangle^2\langle\Delta P\rangle^2-\frac{1}{4}|\leq |\qu|^2$$
and therefore
$$\lim_{|\qu|\longrightarrow 0}|\langle\Delta Q\rangle^2\langle\Delta P\rangle^2|=\frac{1}{2}.$$
Further
$$\frac{1}{2}|\langle[Q, P]\rangle|=\frac{1}{2}r\leq\frac{1}{2}.$$
Therefore
$$\lim_{|\qu|\longrightarrow 0}|\langle\Delta Q\rangle^2\langle\Delta P\rangle^2|\geq \frac{1}{2}|\langle[Q, P]\rangle|.$$
The Heisenberg uncertainty get saturated only in a limit sense, that is in a neighbourhood of zero. We believe this is not due to the way the momentum operator is defined but it is due to the fact that the series $\Ci$ could not be computed explicitly.\\

The following Proposition demonstrate commutativity between quaternions and the right linear operators $\as$ and $\asd$. Further, it plays an important role. This fact is true only for these specific operators and it is not true for general quaternionic linear operators.
\begin{proposition}\label{xAq}\cite{MTI}
	For each $\mathfrak{q}\in\quat$, we have $\mathfrak{q}\cdot {\mathsf{a}} ={\mathsf{a}} \cdot\mathfrak{q}$ and $\mathfrak{q}\cdot {{\mathsf{a}}^\dagger} ={{\mathsf{a}}^\dagger} \cdot\mathfrak{q}$.
\end{proposition}
\subsection{The right quaternionic displacement operator}
On a right quaternionic Hilbert space with a right multiplication we cannot have a displacement operator as a representation for the representation space $\HI^B_r$ \cite{Ad2, Thi2}. However, in \cite{MTI}, we have shown that if we consider a right quaternionic Hilbert space with a left multiplication on it, we can have a displacement operator as a representation for the representation space $\HI^B_r$ with all the desired properties. We shall extract some materials from \cite{MTI} as needed here.\\

\begin{proposition}\label{dis}\cite{MTI} The right quaternionic displacement operator $\D(\qu)=e^{\qu\cdot\asd-\oqu\cdot\as}$ is a unitary, square integrable and irreducible representation of the representation space $\HI^B_r$.
\end{proposition}

	Furthermore, the coherent state $\eta_{\qu}$ is generated from the ground state $\Phi_0$ by the displacement operator $\D(\qu)$,
	\begin{equation}\label{coh_dis}
	\eta_{\qu}=\D(\qu)\Phi_0.
	\end{equation}

\begin{proposition}\label{dis}\cite{MTI} The displacement operator $\D(\qu)$ satisfies the following properties
$$(i)~~\D(\qu)^\dagger\as\D(\qu)=\as+\qu\quad (ii)~~\D(\qu)^\dagger\asd\D(\qu)=\asd+\oqu.$$
\end{proposition}

\subsection{The right quaternionic squeeze operator}
Same reason as for the displacement operator, with a right multiplication on a right quaternionic Hilbert space the squeezed operator cannot be unitary. However, it becomes unitary with a left multiplication on a right quaternionic Hilbert space.
\begin{lemma}\cite{TM}
The operator $A=\pu\cdot (\asd)^2-\overline{\pu}\cdot\as^2$ is anti-hermitian.
\end{lemma}
Therefore by the Baker-Campbell-Hausdorff formula,
$$e^{A}e^{B}e^{-\frac{1}{2}[A, B]}=e^{A+B}$$
we have, for the operator
$$S(\pu)=e^{\frac{1}{2}(\pu\cdot (\asd)^2-\overline{\pu}\cdot\as^2)},$$
$$S(\pu)S(\pu)^\dagger=I_{\HI^B_r}.$$
That is, the operator $S(\pu)$ is unitary and we call this operator the {\em quaternionic squeeze operator}. Further
$$S(\pu)^\dagger=e^{-\frac{1}{2}A}=S(-\pu).$$
If we take
$$K_+=\frac{1}{2}(\asd)^2,\quad K_{-}=\frac{1}{2}\as^2,\quad\text{and}\quad K_0=\frac{1}{2}(\asd\as+\frac{1}{2}I_{\HI^B_r}),$$
Then they satisfy the commutation rules
$$[K_0, K_+]=K_+,\quad [K_0, K_{-}]=-K_{-},\quad\text{and}\quad [K_+,K_{-}]=-2K_0.$$
That is, $K_+, K_{-}$ and $K_0$ are the generators of the $su(1,1)$ algebra and they satisfy the $su(1,1)$ commutation rules. In terms of these operators the squeeze operator $S(\pu)$ can be written as
\begin{equation}\label{sq1}
S(\pu)=e^{\pu\cdot K_+-\opu\cdot K_{-}}.
\end{equation}
The following proposition is the key to compute expectation values and variances of operators. It can be proved using some quaternionic Lie algebraic structures \cite{TM}.
\begin{proposition}\label{Psqop}\cite{TM}
Let $\displaystyle\pu=|\pu|e^{i\theta\sigma(\hat{n})}$ and $N=\asd\as$, the number operator, then the squeeze operator $S(\pu)$ satisfies the following relations
\begin{eqnarray*}
(i)~S(\pu)^{\dagger}\as S(\pu)&=&(\cosh{|\pu|})\as+\left(e^{i\theta\sigma(\hat{n})}\sinh{|\pu|}\right)\cdot\asd.\\
(ii)~S(\pu)^{\dagger}\asd S(\pu)&=&(\cosh{|\pu|})\asd+\left(e^{-i\theta\sigma(\hat{n})}\sinh{|\pu|}\right)\cdot\as.\\
(iii)~S(\pu)^{\dagger}N S(\pu)&=&(\cosh^2{|\pu|})\asd\as+\left(e^{-i\theta\sigma(\hat{n})}\sinh{|\pu|}\cosh{|\pu|}\right)\cdot\as^2\\
&+&\left(e^{i\theta\sigma(\hat{n})}\sinh{|\pu|}\cosh{|\pu|}\right)\cdot(\asd)^2+\sinh^2{|\pu|}\as\asd.
\end{eqnarray*}
\end{proposition}
\subsection{Right quaternionic quadrature operators}
We introduce the quadrature operators analogous to the complex quadrature operators with a left multiplication on a right quaternionic Hilbert space.
\begin{equation}\label{Q-op}
X=\frac{1}{2}(\as+\asd)\quad\text{and}\quad Y=-\frac{i}{2}\cdot (\as-\asd),
\end{equation}
where the quaternion unit $i$ in $Y$ can be replaced by $j, k$ or any $I\in\mathbb{S}$ (see \cite{MTI}).
\begin{proposition}\cite{TM}
The operators $X$ and $Y$ are self-adjoint and $[X,Y]=\frac{i}{2}\cdot I_{\HI^B_r}.$
\end{proposition}
\begin{definition}\cite{Gaz} Let $A$ and $B$ be quantum observables with commutator $[A, B]=i\cdot C$.Then from Cauchy-Schwarz inequality $(\Delta A)(\Delta B)\geq\frac{1}{2}|\langle C\rangle|$. A state will be called squeezed with respect to the pair $(A, B)$ if $(\Delta A)^2$ (or $(\Delta B)^2)<\frac{1}{2}|\langle C\rangle|$.
A state is called ideally squeezed if the equality $(\Delta A)(\Delta B)=\frac{1}{2}|\langle C\rangle|$ is reached together with $(\Delta A)^2$ (or $(\Delta B)^2)<\frac{1}{2}|\langle C\rangle|$.
\end{definition}
We adapt the same definition for quaternionic squeezed states.
\subsection{Right quaternionic pure squeezed states}
A pure squeezed state is produced by the sole action of the unitary operator $S(\pu)$ on the vacuum state. That is, $\eta_\pu=S(\pu)\Phi_0$ are the pure squeezed states. Even through a series form of these states are not necessary to compute the expectation values and variances, we give an expression. Using the BCH formula one can obtain the following. For details see \cite{TM}.
\begin{equation}\label{SCS}
S(\pu)\Phi_0=\eta_\pu=e^{\frac{1}{4}|\pu|^2}\sum_{n=0}^{\infty}e^{n|\pu|^2}\frac{\pu^n\sqrt{(2n)!}}{2^n n!}\cdot\Phi_{2n}.
\end{equation}
Since $S(\pu)$ is a unitary operator, by construction we have
$$\langle\eta_\pu|\eta_\pu\rangle=\langle S(\pu)\Phi_0| S(\pu)\Phi_0\rangle=\langle\Phi_0|\Phi_0\rangle=1.$$
The states $\eta_\pu$ are normalized.
\subsubsection{Expectation values and the variances}
For a normalized state $\eta$ the expectation value of an operator $F$ is $\langle F\rangle=\langle\eta|F|\eta\rangle$. Using Proposition \ref{Psqop} we can obtain the following expectation values . See \cite{TM} for details.
\begin{eqnarray*}
\langle\as\rangle=0,\quad\text{and}\quad \langle\asd\rangle=0.
\end{eqnarray*}
Hence we get
$$\langle X\rangle=\langle\eta_\pu|X|\eta_\pu\rangle=0\quad\text{and}\quad \langle Y\rangle=\langle\eta_\pu|Y|\eta_\pu\rangle=0.$$
Using the same Proposition we readily obtain
\begin{eqnarray*}
\langle\as\asd\rangle&=&\langle\eta_\pu|\as\asd|\eta_\pu\rangle=\cosh^2{|\pu|}\\
\langle\asd\as\rangle&=&\langle\eta_\pu|\asd\as|\eta_\pu\rangle=\sinh^2{|\pu|}\\
\langle\as^2\rangle&=&\langle\eta_\pu|\as^2|\eta_\pu\rangle=\cosh{|\pu|}\sinh{|\pu|}e^{i\theta\sigma(\hat{n})}\\
\langle(\asd)^2\rangle&=&\langle\eta_\pu|(\asd)^2|\eta_\pu\rangle=\cosh{|\pu|}\sinh{|\pu|}e^{-i\theta\sigma(\hat{n})}.
\end{eqnarray*}
Using these expectation values we get
\begin{eqnarray*}
\langle\Delta X\rangle^2\langle\Delta Y\rangle^2
&=&\frac{1}{16}\left\{\mathbb{I}_2+\sinh^2(2|\pu|)\sin^2{(\theta\sigma(\hat{n}))}\right\}
\end{eqnarray*}
This is the quaternionic analogue to the complex case. Since we are in the quaternions, it appears as a $2\times 2$ matrix. Further in the complex case, the product of the variances depends on $r$ and $\theta$ (when $z=re^{i\theta}$). In the quaternion case it depends on all four parameters $r, \theta, \phi$ and $\psi$. Let us write
$$\displaystyle U+iV=e^{-\frac{i}{2}\theta\sigma(\hat{n})}\cdot(X+iY)=e^{-\frac{i}{2}\theta\sigma(\hat{n})}\cdot\as.$$
Then using Proposition \ref{xAq} we can write
\begin{eqnarray*}
S(\pu)^{\dagger}(U+iV)S(\pu)&=&e^{-\frac{i}{2}\theta\sigma(\hat{n})}\cdot S(\pu)^\dagger\as S(\pu)\\
&=&Ue^{|\pu|}+i\cdot Ve^{-|\pu|},
\end{eqnarray*}
with
\begin{eqnarray*}
U&=&\frac{1}{2}(e^{-\frac{i}{2}\theta\sigma(\hat{n})}\cdot\as+e^{\frac{i}{2}\theta\sigma(\hat{n})}\cdot\asd)e^{|\pu|}
\quad\text{and}\\
V&=&\frac{-i}{2}(e^{-\frac{i}{2}\theta\sigma(\hat{n})}\cdot\as-e^{\frac{i}{2}\theta\sigma(\hat{n})}\cdot\asd)e^{|\pu|}.
\end{eqnarray*}
Now it is straight forward that
$\langle\eta_\pu|U|\eta_\pu\rangle=0$, $\langle\eta_\pu|V|\eta_\pu\rangle=0$,
\begin{eqnarray*}
\langle\eta_\pu|U^2|\eta_\pu\rangle
&=&\frac{1}{4}(\cosh{|\pu|}+\sinh{|\pu|})^2\mathbb{I}_2\quad \text{and}\\
\langle\eta_\pu|V^2|\eta_\pu\rangle
&=&\frac{1}{4}(\cosh{|\pu|}-\sinh{|\pu|})^2\mathbb{I}_2.
\end{eqnarray*}
Hence
$$\langle\Delta U\rangle^2 \langle\Delta V\rangle^2=\frac{1}{16}(\cosh^2{|\pu|}-\sinh^2{|\pu|})^2\mathbb{I}_2=\frac{1}{16}\mathbb{I}_2$$
and therefore
\begin{equation}
\langle\Delta U\rangle \langle\Delta V\rangle=\frac{1}{4}\mathbb{I}_2,
\end{equation}
while $\langle\Delta U\rangle \not=\langle\Delta V\rangle$, an exact analogue of the complex case \cite{Gaz}. Hence, the class of ideally squeezed states contains the set of quaternionic pure squeezed states.\\
Using the relation (iii) in Proposition \ref{Psqop} we obtain the mean photon number
$$\langle N\rangle=\langle\eta_\pu| N|\eta_\pu\rangle=\langle\Phi_0|S(\pu)^\dagger NS(\pu)\Phi_0\rangle=\sinh^2{|\pu|}\I_2.$$
and
\begin{eqnarray*}
\langle N^2\rangle&=&\langle\Phi_0|S(\pu)^\dagger N S(\pu)S(\pu)^\dagger N S(\pu)\Phi_0\rangle
=3\sinh^4{|\pu|}+2\sinh^2{|\pu|}\I_2.
\end{eqnarray*}
Hence the variance is
$$\langle \Delta N\rangle^2=\langle N^2\rangle-\langle N\rangle^2=2\sinh^2{|\pu|}(1+\sinh^2{|\pu|})\I_2.$$
The photon number variance is also described by Mandel's Q-parameter. The Mandel parameter is \cite{Gaz, Lo}
\begin{eqnarray*}
Q_M&=&\frac{\langle \Delta N\rangle^2}{\langle N\rangle}-1
=(1+2\sinh^2{|\pu|})\I_2=2\langle N\rangle+\I_2.
\end{eqnarray*}
Since $Q_M>0$ (as a positive definite matrix) the photon number probability distribution is super-Poissonian.
($Q_M=0$ Poissonian and $Q_M<0$ sub-Poissonian).

\subsection{Right quaternionic squeezed states}
According to Prop. \ref{lft_mul}(f), a basis vector satisfies $\qu\cdot\Phi_n=\Phi_n\qu$, therefore we write the canonical CS as
$$\eta_\qu=\D(\qu)\Phi_0=e^{-|\qu|^2/2}\sum_{n=0}^{\infty}\Phi_n\frac{\qu^n}{\sqrt{n!}}.$$
Let $\displaystyle S(\pu)\Phi_n=\Phi_n^{\pu}$, where the set $\{\Phi_n~~|~~n=0,1,2,\cdots\}$ is the basis of the Fock space of regular Bargmann space $\HI^{B}_r$. Since $S(\pu)$ is a unitary operator, the set $\{\Phi_n^{\pu}~~|~~n=0,1,2,\cdots\}$ is also form an orthonormal basis for $\HI^{B}_r$.
Now the squeezed states are
\begin{equation}\label{sqstates}
\eta_{\qu}^{\pu}=S(\pu)\D(\qu)\Phi_0=S(\pu)\eta_{\qu}
=e^{-|\qu|^2/2}\sum_{n=0}^{\infty}\Phi_n^{\pu}\frac{\qu^n}{\sqrt{n!}}.
\end{equation}
Since the canonical CS are normalized, that is $\langle\eta_\qu|\eta_\qu\rangle=1$, and the squeeze operator $S(\pu)$ is unitary, we have
$$\langle\eta_\qu^\pu|\eta_\qu^\pu\rangle=\langle S(\pu)\eta_\qu|S(\pu)\eta_\qu\rangle=\langle\eta_\qu|\eta_\qu\rangle=1.$$
That is, the squeezed states are normalized. The dual vector of $|S(\pu)\eta_\qu\rangle$ is $\langle\eta_\qu S(\pu)^\dagger|$. Therefore, from the resolution of the identity of the canonical CS,
$$\int_\quat |\eta_\qu\rangle\langle\eta_\qu| d\zeta(r,\theta,\phi,\psi)=I_{\HI^B_r}$$
we get
$$\int_\quat |S(\pu)\eta_\qu\rangle\langle\eta_\qu S(\pu)^\dagger| d\zeta(r,\theta,\phi,\psi)=S(\pu)I_{\HI^B_r}S(\pu)^\dagger=I_{\HI^B_r}.$$
That is the squeezed states satisfy the resolution of the identity,
$$\int_\quat |\eta_\qu^\pu\rangle\langle\eta_\qu^\pu| d\zeta(r,\theta,\phi,\psi)=I_{\HI^B_r}.$$

\begin{remark}
Since the operators $\D(\pu)$ and $S(\qu)$ are unitary operators the states $\D(\pu)S(\qu)\Phi_0$ are normalized.
\end{remark}
Since quaternions do not commute the expectation vales cannot be computed. For example, if we combine the Propositions \ref{dis} and \ref{Psqop}, when $\pu=|\pu|e^{i\theta\sigma(\hat{n})}$ let $I_\pu=e^{i\theta\sigma(\hat{n})}$,
\begin{eqnarray*}
\D(\qu)^\dagger S(\pu)^{\dagger}\as S(\pu) \D(\qu) &=&
\D(\qu)^{\dagger} \left[(\cosh{|\pu|})\as+I_\pu\sinh{|\pu|}\cdot\asd\right]\D(\qu)\\
&=&\cosh{|\pu|}\D(\qu)^\dagger\as\D(\qu)+\sinh{|\pu|}\D(\qu)^\dagger I_\pu\cdot\asd\D(\qu).
\end{eqnarray*}
Since $\D(\qu)^\dagger I_\pu\cdot\asd\D(\qu)\not=I_\pu\cdot\D(\qu)^\dagger\asd\D(\qu)$, the above expression cannot be computed. In fact, there is no known technique in quaternion analysis to get a closed form for the expression $\D(\qu)^\dagger I_\pu\cdot\asd\D(\qu)$. However, since elements in a quaternion slice commute, if we consider squeezed states in a quaternion slice then the computations can carry forward.
\section{Squeezed states on a quaternion slice}
Since elements in a quaternion slice $\C_I$ commute we can obtain all the desired results. The states $\D(\qu)S(\pu)\Phi_0$ are called the two photon coherent states \cite{Yuen, Rodney}. The states $\D(\qu)S(\pu)\Phi_0$ are called the squeezed coherent states \cite{Rodney} pp. 207. In the following we briefly see some relations.
\subsection{Two photon coherent states}
Let $\pu,\qu\in\C_I$, then the two photon coherent states are defined as $\eta_{\qu}^\pu=\D(\qu)S(\pu)\Phi_0$ \cite{Yuen}. Let
\begin{eqnarray*}
\pu&=&|\pu|e^{I\theta_{\pu}}=|\pu|I_{\pu}=|\pu|(\cos\theta_\pu+I\sin\theta_\pu)\quad\text{and}\\
\qu&=&|\qu|e^{I\theta_{\qu}}=|\qu|I_{\qu}=|\qu|(\cos\theta_\qu+I\sin\theta_\qu).
\end{eqnarray*}
With these notations we obtain the following.
\begin{proposition} The operators $S(\pu)$ and $\D(\qu)$ satisfies the following relations.
\begin{eqnarray*}
\D(\qu)^\dagger S(\pu)^\dagger\as S(\pu)\D(\qu)&=&\cosh{|\pu|}\as\mathbb{I}_2+I_\pu\sinh{|\pu|}\cdot\asd
+\cosh{|\pu|}\qu\I_2+I_\pu\sinh{|\pu|}\oqu\\
\D(\qu)^\dagger S(\pu)^\dagger\asd S(\pu)\D(\qu)&=&\cosh{|\pu|}\asd\mathbb{I}_2+\overline{I}_\pu\sinh{|\pu|}\cdot\as
+\cosh{|\pu|}\oqu\I_2+\overline{I}_\pu\sinh{|\pu|}\qu,\\
\D(\qu)^{\dagger}S(\pu)^{\dagger}N S(\pu)\D(\qu)&=&\cosh^2{|\pu|}(N+\qu\cdot\asd+\oqu\cdot\as+|\qu|^2)\\
&+&\frac{1}{2}\overline{I}_\pu\sinh{(2|\pu|)}\cdot(\as^2+2\qu\cdot\as+\qu^2)\\
&+&\frac{1}{2}I_\pu\sinh{(2|\pu|)}\cdot((\asd)^2+2\oqu\cdot\asd+\oqu^2)\\
&+&\sinh^2{|\pu|}(\as\asd+\oqu\cdot\as+\qu\cdot\asd+|\qu|^2).
\end{eqnarray*}
\end{proposition}
\begin{proof} Proof is straight forward from the results of the Propositions \ref{Psqop} and \ref{dis}.
\end{proof}
Using these relations all the desired expectation values and variances can be obtained. See for details \cite{TM}.

\subsection{Squeezed coherent states}
The squeezed coherent states are defined as $\eta_\pu^{\qu}=\D(\qu)S(\pu)\Phi_0$ \cite{Rodney, Gaz}. We briefly provide some formulas for these states. Once again we are in a quaternion slice $\C_I$ and $\pu$ and $\qu$ are as in the previous section.
\begin{proposition}The operators $\D(\qu)$ and $S(\pu)$ satisfy the following relations.
\begin{eqnarray*}
S^{\dagger}(\pu)\D(\qu)^{\dagger}\as\D(\qu)S(\pu)&=&\cosh{|\pu|}~\as+I_\pu\sinh{\pu}~\asd+\qu\\
S^{\dagger}(\pu)\D(\qu)^{\dagger}\asd\D(\qu)S(\pu)&=&\cosh{|\pu|}~\asd+\overline{I}_\pu\sinh{\pu}~\as+\oqu\\
S^{\dagger}(\pu)\D(\qu)^{\dagger}\asd\as\D(\qu)S(\pu)&=&\cosh^2{|\pu|}~\asd\as+\frac{1}{2}I_\pu\sinh(2|\pu|)~(\asd)^2
+\qu\cosh{|\pu|}~\asd\\
&+&\frac{1}{2}\overline{I}_\pu\sinh(2|\pu|)~\as^2+\sinh^2{|\pu|}~\as\asd+\overline{I}_\pu\qu\sinh{|\pu|}~\as\\
&+&\oqu\cosh{|\pu|}~\as+I_\pu\oqu\sinh{|\pu|}~\asd+|\qu|^2.
\end{eqnarray*}
\end{proposition}
\begin{proof}
Proof is straight forward from Propositions \ref{Psqop} and \ref{dis}.
\end{proof}
Once again using these relations all the required expectation values and variances can be obtained \cite{TM}.
\section{Right quaternionic fermionic coherent states for one mode}
The material in this section has not appeared in the literature, however, computationally, it is somehow similar to the squeezed states. Once again using a left multiplication on a right quaternionic Hilbert space we present quaternionic fermionic coherent states for one mode. These states are superpositions of number states $|n\rangle$ but only $n=0$ or $n=1$. These states are used in atomic and nuclear physics \cite{Gaz}.\\
Since $\{|0\rangle, |1\rangle\}$ is the basis, according to Proposition \ref{lft_mul} (f), these vectors commute with quaternions. The action of the creation and annihilation operators are
$$\as |0\rangle=0,\quad \as |1\rangle=|0\rangle,\quad \asd |0\rangle=|1\rangle,\quad \asd |1\rangle=0.$$
Therefore Proposition \ref{xAq} is valid for these operators. The creation and annihilation operators for a Fermionic mode has to obey the following commutation rules:
\begin{equation}\label{fer}
[\as, \asd]_+:=\as\asd+\asd\as=I_{\HIB},\quad [\as, \as]_+=0,\quad\text{and}\quad [\asd, \asd]_+=0.
\end{equation}
Note that $(\asd)^2=0$ and $\as^2=0$.
Any system having such commutation rules has the group $SU(2)$ as the dynamical group. The generators of this group are $\{\asd, \as, \asd\as-\frac{1}{2} \}$ \cite{Gaz}. Using the commutation rules \ref{fer}, we can easily obtain the usual commutation rules
\begin{equation}\label{fer_usual}
[\asd, \as]=2(\asd\as-\frac{1}{2}),\quad [\asd\as-\frac{1}{2}, \as]=-\as,\quad [\asd\as-\frac{1}{2}, \asd]=\asd.
\end{equation}
Define the fermionic states, with $\qu\in\quat$, such that
\begin{equation}\label{states}
\eta_0=e^{\qu\cdot\asd-\oqu\cdot\as}|0\rangle\quad\text{and}\quad \eta_1=e^{\qu\cdot\asd-\oqu\cdot\as}|1\rangle
\end{equation}
Now with the aid of Proposition \ref{xAq} we can see that
$$(\qu\cdot\asd-\oqu\cdot\as)^2=-|\qu|^2 I_{\HIB}.$$
From this relation we can easily obtain the following.
\begin{eqnarray*}
e^{\qu\cdot\asd-\oqu\cdot\as}&=&\sum_{n=0}^{\infty}\frac{(\qu\cdot\asd-\oqu\cdot\as)^n}{n!}\\
&=&\sum_{n=0}^{\infty}\frac{|\qu|^{2n}}{(2n)!}~I_{\HIB}+
\frac{\qu}{|\qu|}\sum_{n=0}^{\infty}\frac{(-1)^n|\qu|^{2n+1}}{(2n+1)!}\cdot\asd-
\frac{\oqu}{|\qu|}\sum_{n=0}^{\infty}\frac{(-1)^n|\qu|^{2n+1}}{(2n+1)!}\cdot\as\\
&=&\cos{|\qu|}~I_{\HIB}+\frac{\qu}{|\qu|}\sin{|\qu|}\cdot\asd-
\frac{\oqu}{|\qu|}\sin{|\qu|}\cdot\as.
\end{eqnarray*}
Hence we have the {\em right quaternionic fermionic states} for one mode as
\begin{eqnarray*}
\eta_0&=&\cos{|\qu|}~|0\rangle \I_2+\frac{\qu}{|\qu|}\sin{|\qu|}\cdot|1\rangle\\
\eta_1&=&\cos{|\qu|}~|1\rangle\I_2-\frac{\oqu}{|\qu|}\sin{|\qu|}\cdot|0\rangle.
\end{eqnarray*}
By choosing appropriate choice for $\qu$ it can be matched to the complex case. For example if we take $\qu\in\quat$ as $\displaystyle\qu=|\qu|e^{i\theta\sigma(\hat{n})}$ then we get
\begin{eqnarray*}
\eta_0&=&\cos{|\qu|}~|0\rangle \I_2+\sin{|\qu|}e^{i\theta\sigma(\hat{n})}\cdot|1\rangle\\
\eta_1&=&\cos{|\qu|}~|1\rangle\I_2-\sin{|\qu|}e^{i\theta\sigma(\hat{n})}\cdot|0\rangle,
\end{eqnarray*}
which is the quaternionic analogue of the complex case.
\section{Conclusion}
  With the aid of a left multiplication defined on a right quaternionic Hilbert space we have defined unitary squeeze operator. Using this operator pure squeezed states have been obtained, with all the necessary properties, analogous to the complex case. Using the displacement operator and the squeeze operator we have defined squeezed states. However, the noncommutativity of quaternions prevented us in getting desired results. It is just a technical issue, but there is no known technique to overcome this difficulty. The only way out of this difficulty is to consider quaternionic slices. We have briefly defined squeezed states on quaternion slices and provided some necessary formulas. We constructed quaternionic fermionic coherent states in complete analogy with their complex counterpart.

 Squeezed states have several applications, particularly in coding and transmission of information through optical devices.  The fermionic states gained application in atomic and nuclear physics. These aspects are well explained for example in \cite{Ali, Gaz, Yuen} and the many references therein. Since we have used the matrix representation of quaternions, the squeezed states and the fermionic states obtained in this note appear as matrix states. Further these states involve all four variables of quaternions. These features may give advantage in applications.

%
\end{document}